\documentclass[11pt]{article}

\usepackage{lineno,amsmath,amsthm,epsfig,hyperref,float,mathpazo,enumerate,verbatim,lmodern,xparse,tikz,caption,subcaption,color,graphicx}
\usepackage[super]{nth}

\newcommand{\fv}[1]{\textcolor{blue}{ {\textbf{(Full Version:}
#1\textbf{) }}}}
\renewcommand{\fv}[1]{#1}

\newcommand{\av}[1]{\textcolor{blue}{ {\textbf{(Extended abstract Version:}
#1\textbf{) }}}}
\renewcommand{\av}[1]{#1}
\renewcommand{\av}[1]{}

\usepackage{}

\mathchardef\mhyphen="2D 

\newtheorem{theorem}{Theorem}

\newtheorem{lemma}[theorem]{Lemma}

\newtheorem{definition}[theorem]{Definition}

\newtheorem{conjecture}[theorem]{Conjecture}

\newenvironment{proof-sketch}{\noindent{\bf Sketch of Proof:}\hspace*{1em}}{\qed\bigskip}
\newenvironment{proof-idea}{\noindent{\bf Proof Idea:}\hspace*{1em}}{\qed\bigskip}
\newenvironment{proof-of-lemma}[1]{\noindent{\bf Proof of Lemma #1:}\hspace*{1em}}{\qed\bigskip}
\newenvironment{proof-of-proposition}[1]{\noindent{\bf Proof of Proposition #1:}\hspace*{1em}}{\qed\bigskip}
\newenvironment{proof-attempt}{\noindent{\bf Proof Attempt:}\hspace*{1em}}{\qed\bigskip}

\newcommand{\tensor}{\otimes}

\newcommand{\bra}[1]{\langle #1|}
\newcommand{\ket}[1]{|#1\rangle}
\newcommand{\braket}[2]{\langle #1|#2\rangle}
\newcommand{\ketbra}[2]{\ket{#1}{\bra{#2}}}

\renewcommand{\P}{\textsf{P}}

\newcommand{\NP}{\textsf{NP}}
\newcommand{\NPC}{\textsf{NP}\mhyphen\textsf{complete}}
\newcommand{\NPH}{\textsf{NP-Hard}}

\newcommand{\QMA}{\textsf{QMA}}

\newcommand{\QMAC}{\textsf{QMA-complete}}
\newcommand{\QMAH}{\textsf{QMA-hard}}
\newcommand{\QMAo}{\textsf{QMA$_1$}}
\newcommand{\QMAoC}{\textsf{QMA$_1$-complete}}
\newcommand{\QMAoH}{\textsf{QMA$_1$-hard}}

\newcommand{\lh}{\textsc{local hamiltonian}}
\newcommand{\klh}{k \mhyphen \textsc{local hamiltonian}}

\newcommand{\sat}{\textsc{sat}}
\newcommand{\ksat}{k\mhyphen\textsc{sat}}
\newcommand{\kgsat}[1][\gamma]{(k,#1)\mhyphen\textsc{sat}}
\newcommand{\kssat}[1][s]{(k,#1)\mhyphen\textsc{sat}}

\newcommand{\qsat}{\textsc{qsat}}
\newcommand{\kqsat}{k\mhyphen\textsc{qsat}}

\newcommand{\ksqsat}[1][s]{(k,#1)\mhyphen\textsc{qsat}}

\newcommand{\ql}{l^{*}}
\newcommand{\qf}{f^{*}}
\newcommand{\qg}{g^{*}}
\newcommand{\qh}{h^{*}}

\newcommand{\onote}[1]{\textcolor{blue}{ {\textbf{(Or:}
#1\textbf{) }}}}
\renewcommand{\onote}[1]{}

\usepackage{color}
\definecolor{gray}{rgb}{0.5,0.5,0.5}

\title{An Almost Sudden Jump in Quantum Complexity
\av{%
\\ \vspace{2 mm} {\large Extended abstract\footnote{A full version preprint  with the same title is available on the arXiv.}}
}%
}

 \author{Or Sattath
 \thanks{School of Computer Science and Engineering,
 The Hebrew University, Jerusalem, Israel. }
 }

\begin{document}

\maketitle
\begin{abstract}
The Quantum Satisfiability problem ($\qsat$) is the generalization of the canonical $\NPC$ problem - Boolean Satisfiability. 
$\ksqsat$ is the following variant of the problem: given a set of projectors of rank $1$, acting non-trivially on $k$ qubits out of $n$ qubits, such that each qubit appears in at most $s$ projectors, decide whether there exists a quantum state in the null space of all the projectors. 
Let $\qf(k)$ be the maximal integer $s$ such that every $\ksqsat$ instance is satisfiable. Deciding $\ksqsat[\qf(k)]$ is computationally easy: by definition the answer is ``satisfiable''. But, by relaxing the conditions slightly, we show that $\ksqsat[\qf(k)+2]$ is $\QMAoH$, for $k \geq 15$. This is a quantum analogue of a classical result by Kratochv{\'\i}l et al. \cite{kratochvil1993one}.
We use the term ``an \emph{almost} sudden jump'' to stress that the complexity of $\ksqsat[\qf(k)+1]$ is open, where the jump in the classical complexity is known to be sudden. 

We present an implication of this finding to the quantum PCP conjecture, arguably one of the most important open problems in the field of Hamiltonian complexity. Our implications impose constraints on one possible way to refute the quantum PCP.

\end{abstract}

\section{Introduction}
The quantum Satisfiability problem, introduced by Bravyi~\cite{bravyi2006efficient}, generalizes the Boolean Satisfiability problem ($\sat$) to the quantum setting. The $3 \mhyphen \qsat$ problem \fv{(see Definition \ref{def:qsat})} is $\QMAoC$~\cite{gosset2013quantum} \fv{(improving previous results~\cite{bravyi2006efficient,eldard2008quantum})} and $2 \mhyphen \qsat \in \P$~\cite{bravyi2006efficient}.
$\QMAo$ is a quantum generalization of the class $\NP$.It differs from the more familiar class $\QMA$ in having a one-sided error.  
\fv{%
See Ref.~\cite{aharonov2002quantum} for a survey of $\QMA$,~\cite{osborne2012hamiltonian} for a broader discussion, and Ref.~\cite{bookatz2012qma} for a thorough description of $\QMAoC$ and $\QMAC$ problems. 
}

A $\qsat$ instance, in which each projector acts non-trivially on exactly $k$ qubits, has rank $r$ on these qubits, and each qubit appears in at most $s$ projectors, is called a $(k \text{ local},\text{rank } r, \text{degree }s )$ instance. We define $\ksqsat$ to be the $\qsat$ problem restricted to $(k \text{ local},\text{rank } 1, \text{degree }s)$ instances \fv{(see Definition~\ref{def:ksqsat})}. We have removed the rank-1 notation, since this is the only case of interest here. The reason for focusing on rank-1 instances is the analogy with $\ksat$: in a $\ksat$formula, each clause excludes one out of $2^{k}$ configurations of the relevant variables; In rank-1 $\kqsat$,  each projector excludes one dimension out of the $2^{k}$ dimensions. 

The following definition is essential for the rest of this work.
\begin{equation}
\qf(k) \equiv \max\{s\in \mathbb{N}\ | \text{ all } \ksqsat \text{ instances are satisfiable} \} .
\label{def:qfk}
\end{equation}
Our main result is the following:


\begin{theorem} For $k \geq 15$, $\ksqsat[\qf(k)+2]$ is $\QMAoC$.
\label{thm:main_theorem}
\end{theorem}

We say that there is a jump in complexity, because deciding $\ksqsat$ instances with $s \leq \qf(k)$ is computationally easy (by Eq.~\eqref{def:qfk} the answer is always ``satisfiable''); yet, deciding instances with degree $s$ bigger by 2 is $\QMAoH$.

\av{The above theorem is a quantum analogue of the following classical theorem:}
\fv{
The above theorem is a quantum analogue of a similar classical theorem, which requires the parallel definitions.
A $\sat$ instance in Conjunctive Normal Form (CNF), in which each clause contains exactly $k$ different variables, and each variable appears in at most $s$ clauses, is called a $(k \text{ CNF},\text{degree }s)$ formula. Let $\kssat$ be the problem $\sat$ restricted to $(k\text{ CNF}, \text{degree }s)$ formulas.
The classical analogue of $\qf(k)$ is defined as follows:

\[f(k)=\max\{s \in \mathbb{N} | \text{all } \kssat \text{-formulas are satisfiable} \}.\] 
}

\begin{theorem}[\cite{kratochvil1993one}] For $k \geq 3$, $\kgsat[f(k)+1]$ is $\NPC$.
\label{thm:kratochvil}
\end{theorem}

We use the term ``an almost sudden jump''\footnote{The term ``sudden jump in complexity'' was coined by Gebauer, Szab{\'o} and Tardos in Ref.~\cite{gebauer2010local}.} in the title to stress that the complexity of $\ksqsat[\qf(k)+1]$ is still open (unlike the truly sudden jump in the classical case).
\fv{%
We conjecture that $\ksqsat[\qf(k)+1]$ is $\QMAoC$, and that the reason for the divergence between the classical and quantum case is due to the technical limitations of our proof technique.
}

\paragraph{An overview of the proof of Theorem~\ref{thm:main_theorem}.} 

It is already known that $\kqsat \in \QMAo$ for any $k$~\cite{bravyi2006efficient}, which implies that also $\ksqsat \in \QMAo$.
For reasons which will be elucidated shortly, we wish to start the hardness reduction with a $\qsat$ instance $Q$ which is promised to have the smallest possible degree and locality; this is achieved by the hardness result of $\qsat$ on a line~\cite{aharonov_power_2007}.\footnote{Note that the original $\QMA$ hardness result by Bravyi~\cite{bravyi2006efficient} does not provide an upper bound on the degree, and therefore cannot be used.} 

Next, every $\Pi \in Q$ is replaced with $\Pi' = \Pi \tensor  \ket{0}\bra{0}$, where the last qubit is denoted the dummy qubit, and a new dummy qubit is used for every projector. If the dummy qubit could be enforced to be in the $\ket{0}$ state, the satisfiability of the instance would not be affected. In order to achieve this, we use a $\qsat$ \emph{enforcing gadget}, which has the following properties: (i) the only way to satisfy the gadget is by having the dummy qubit in the $\ket{0}$ state, (ii) the degree of the gadget is at most $\qf(k) + 2$, (iii) the degree of the dummy qubit is $2$. Replacing $\Pi$ with $\Pi'$, and adding the enforcing gadget, has the following three effects: the satisfiability of the instance is unaffected; the locality increases by 1; the degree changes to the maximum between the original degree and $\qf(k) + 2$.

By repeating this argument, we can increase the initial locality to some large enough $k$, while keeping the degree at most the maximum between the original degree and $\qf(k) + 2$.
Note that $\qf(k)$ grows roughly exponentially in $k$ (\fv{see Theorem~\ref{thm:bound_qf}}\av{see Eq.~\eqref{eq:bounds_f_qf}}).
For $k$ large enough,  $\qf(k)+2$ will be larger than the initial degree, hence the final degree is $\qf(k)+2$. 
Therefore, this transformation is useful if the original instance has a small locality and degree which explains our initial choice of these parameters.

\fv{
A comparison between the proofs of Theorem~\ref{thm:main_theorem} and Theorem~\ref{thm:kratochvil} is given at the end of Section~\ref{sec:proof}.
}

\paragraph{Known results about $\qf(k)$ and $f(k)$.}
\fv{%
 Every $\kssat$ formula is also a $\ksqsat$ instance, therefore 
\begin{equation}
\qf(k) \leq f(k).
\end{equation}
The Lov\'asz Local Lemma~\cite{erdos1975problems,alon2004probabilistic} was used to lower bound $f(k)$.
\begin{theorem}[\cite{kratochvil1993one}]
$ f(k) \geq \lfloor \frac{2^{k}}{ek} \rfloor  $.
\label{thm:bound_fk}
\end{theorem}

One of the merits of the Quantum Lov\'asz Local Lemma is that it provides the same lower bound in the quantum setting.
 \begin{theorem}[\cite{ambainis2010quantum}]
$\qf(k)\geq \lfloor \frac{2^{k}}{ek} \rfloor $.
\label{thm:bound_qf}
\end{theorem}

Gebauer et al. improved the classical lower bound using the Lopsided Lov\'asz Local Lemma~\cite{erdos1991lopsided}, and also proved that it is asymptotically tight.
\begin{theorem}[\cite{gebauer2010local}]
$  \lfloor \frac{2^{k+1}}{e(k+1)} \rfloor \leq f(k) \leq \frac{2^{k+1}}{ek}\left(1 + O(\frac{1}{\sqrt{k}}) \right) $.
\label{thm:gebauer}
\end{theorem}
Although these bounds are satisfying asymptotically, it is unknown whether the functions $f(k)$ and $\qf(k)$ are computable~\cite{hoory2005computing}. 


All of these results, excluding Gebauer et al.'s lower bound, can be combined in the following way.
}%
\av{%
The main results known about $\qf(k)$ and $f(k)$ are summarized in the following equation. See the full version for more details.
 }

\begin{equation}
\label{eq:bounds_f_qf}
\lfloor \frac{2^{k}}{ek}  \rfloor  \leq \qf(k) \leq f(k) \leq \frac{2^{k+1}}{ek}\left(1 + O(\frac{1}{\sqrt{k}}) \right) 
\end{equation}
\fv{
The question whether $\qf(k) = f(k)$ is addressed in the open questions, at the end of this section.

Similar functions to $f(k)$ and $\qf(k)$ are $l(k)$ and $\ql(k)$. The \emph{neighborhood} of a clause is the number of clauses that share variables with it. $l(k)$ is the maximal integer $m$ such that every $\ksat$ instance in which the neighborhood of each clause is at most $m$, is satisfiable. Similarly, $\ql(k)$ is defined in the same manner in the quantum setting. An analogous result to Theorem~\ref{thm:kratochvil} where $l(k)$ takes the role of $f(k)$ was given in~\cite{gebauer2009lovasz}. Using a similar construction\footnote{The only change in the proof is that $R$ is replaced with a $\ksqsat[\ql(k)+1 \text{ neighborhood}]$ instance. $R$ is defined in Section~\ref{sec:proof}.} to the one used for proving Theorem~\ref{thm:main_theorem}, it can be shown that $\ksqsat[\ql(k)+3 \text{ neighborhood}]$ is $\QMAoC$, for $k \geq 11$.
}
\paragraph{Implication for the quantum PCP conjecture.}
Theorem \ref{thm:main_theorem} has an unexpected connection to the quantum Probabilistically Checkable Proof (qPCP) conjecture. 

\fv{%
We start by describing the connection imprecisely, and elaborate and discuss the caveats later. 
}%
\av{%
The following crude discussion presents the intuition, and the accurate details are presented in the full version.
}%

The qPCP conjecture states that approximating $e_{0}\equiv \lambda_{0}/m$  (the minimum eigenvalue of a $\kqsat$ instance, normalized by the number of projectors\fv{, see Definition~\ref{def:qsat}}), up to an additive constant, is $\QMAoH$.
\fv{%
A survey of the qPCP conjecture appears in~\cite{aharonov2013guest}.
}%




%

Brand\~ao and Harrow proved\footnote{under the reasonable assumption that $\NP \neq \QMAo$} that it is not $\QMAoH$ to approximate $e_{0}$ up to any  additive constant, for $\ksqsat$ instances when $s$ is super constant, that is $s = \Omega(1)$. We also know that deciding $\ksqsat$ when $s$ is small enough is not $\QMAoH$. 

These two facts suggest a way to refute the qPCP conjecture for a given $k$, by proving that there exists a constant $s_{0}$ for which:
\begin{enumerate}[(a)]
\item
Approximating $e_{0}$  for $\ksqsat$ instances up to any additive constant for $s > s_{0}$ is not $\QMAoH$.
\item Deciding $\ksqsat$ for $ s \leq s_{0}$ is not $\QMAoH$. 
\end{enumerate}



Though (a) is a strengthened version of the result by Brand\~ao and Harrow, it is not too strong since refuting the qPCP straightforwardly implies it.

Note that (a) refers to $s > s_{0}$ and (b) refers to $s \leq s_{0}$; therefore, their combination refutes the qPCP for the given k. 

The higher the value for $s_{0}$, the easier it is to prove (a). This would lead us to seek high values of $s_{0}$. Clearly, (b) holds for $s_{0} \leq \qf(k)$ from Eq.~\eqref{def:qfk}. But what about $s_{0}$ which is higher? Theorem~\ref{thm:main_theorem} contradicts (b)  for $s_{0} \geq \qf(k)+2$ for $k\geq 15$, therefore, further increasing $s_{0}$, would not succeed.

In other words, refuting the qPCP conjecture by exploiting the non-hardness of deciding $\ksqsat$ for low values of $s$ would not work for $s > \qf(k) + 1$. 

\fv{It is an interesting fact that in the classical setting, Berman, Karpinski and Scott proved a PCP theorem for $\textsc{max}\  (3,f(3)+1) \mhyphen \sat$ \cite{berman2003approximation}. \onote{how is this related to the previous paragraph?}
}

\fv{

We now describe the connection with more rigor and detail. 

\begin{conjecture}[The regular qPCP for $\qsat$] There exist universal constants $k , \epsilon > 0$ for which deciding whether a rank-1 regular $\kqsat$  instance $Q$ is satisfiable (in other words, $e_{0}(Q) = 0$), or $e_{0}(Q) \geq \epsilon$ is $\QMAoH$.
\label{con:qpcp_qsat}
\end{conjecture}

\begin{conjecture}[The regular qPCP for local Hamiltonians] There exist universal constants $k ,  \epsilon_{0} < \epsilon_{1}  $, for which deciding whether a regular $\klh$ instance $H$, satisfies $e_{0}(H) \leq \epsilon_{0}$ or $e_{0}(H) \geq \epsilon_{1}$ is $\QMAoH$.
\label{con:qpcp_lh}
\end{conjecture}
The regularity assumption, stated in both conjectures, is that the degree of all the qubits is equal (see Definition~\ref{def:regularity_degree}).
The qPCP conjecture (see, for example, \cite{aharonov2013guest}) is usually defined as in Conjecture~\ref{con:qpcp_lh}, but without the regularity assumption. The first caveat in our argument is that the regularity assumption is non-standard: it is added in order to make the result by Brand\~ao and Harrow applicable. 

Clearly, Conjecture~\ref{con:qpcp_qsat} implies Conjecture~\ref{con:qpcp_lh}. The qPCP for Local Hamiltonians for $k > 2$ implies it for $k=2$~\cite{bravyi2008quantum}. A parallel result is not known for the qPCP for $\qsat$ conjecture. Such a result would simplify our statement: in order to refute the qPCP, one needs to provide an $s_{0}(k)$ with the desired properties for each $k \geq 2$. A parallel result would imply that one only needs to find such an $s_{0}$ for $k=3$ (the smallest $k$ for which the problem is $\QMAoC$). 

Brand\~ao and Harrow proved the following.
\begin{theorem}[Adapted from \cite{brandao2013product}]  For regular $2 \mhyphen \lh$ instances, with super constant degrees, in which each pair of qubits appears in at most one term in the Hamiltonian, $e_{0}$ can be approximated up to any additive constant in $\NP$.
 \end{theorem}
 This result refutes both the strong (for all $k$) and the weak (for $k=2$) regular qPCP conjectures for super constant degrees. 

 The requirement that each pair of qubits appears in at most one term is because, otherwise, the result would be much stronger, and refute the qPCP for local Hamiltonian: we could replace each term $H_{i}$ by $\ell = \Omega(1)$ copies of it. The new Hamiltonian $H' = \ell H$, and therefore $e_{0}(H)=e_{0}(H')$, yet, the degree of $H'$ is $\Omega(1)$. Therefore, without this requirement, the qPCP conjecture would be refuted, since $e_{0}(H)$ could be approximated up to any constant in $\NP$.
 
}

\av{
\paragraph{Discussion.}
The reduction in the main proof is non-constructive in a very strong sense. The exact value of $\qf(k)$ is not known to be computable. How can the reduction produce an instance with degree $\qf(k)+2$ when $\qf(k)$ is not even known to be computable?
The key idea is that it is guaranteed that there exists a non satisfiable $\ksqsat[\qf(k)+1]$ instance, which is manipulated in a black box manner in the reduction. This is the first non-constructive quantum reduction the author is aware of.
}
\fv{%
\paragraph{Discussion.} We find Theorem~\ref{thm:main_theorem} interesting due to the following two reasons. The first is that the proof is non-constructive in a very strong sense (in the same sense as the classical construction in Theorem~\ref{thm:kratochvil}). The exact value of $\qf(k)$ is not known to be computable, yet the reduction defines an instance with degree at most $\qf(k) + 2$. How can the reduction produce an instance with degree $\qf(k)+2$ without being able to compute $\qf(k)$?
The key idea is that it is guaranteed that there exists a non satisfiable $\ksqsat[\qf(k)+1]$ instance, which is manipulated in a black box manner in the reduction. This is the first non-constructive quantum reduction the author is aware of.

The second reason is the properties of the construction. Generalizing known results for the class $\NP$ to the class $\QMA$ is typically a difficult task. The fact that a generalization in this case is possible, may lead to using a similar approach for other tasks. The quantum generalization uses a completely different approach for the first step of the proof (degree reduction) and a modification of the second step (modification of the enforcing gadget due to entanglement). A detailed comparison between the quantum and classical proofs (Theorem~\ref{thm:main_theorem} and Theorem~\ref{thm:kratochvil} respectively) appears at the end of Section~\ref{sec:proof}.
}
\paragraph{Open problems.}
\av{%
Since every $\ksqsat$ instance is also a $\kssat$ instance, $\qf(k) \leq f(k)$. Is $\qf(k) = f(k)$? Interestingly, all the techniques to prove a lower bound on $f(k)$ can be generalized to the quantum setting, \emph{except one}: the Lopsided Lov\'asz Local Lemma. This technique seems hard to generalize because it takes into account the signs of the literals, and it is unclear what could take that role in the quantum setting. This may suggest a way of proving that $\qf(k)<f(k)$: find a non-satisfiable $\ksqsat$ instances in the regime $s$ where the (classical) Lopsided Lov\'asz Local Lemma is the only applicable technique to prove lower bounds on $f(k)$. 

Can the jump be made truly sudden (i.e. to prove that $\ksqsat[\qf(k)+1]$ is $\QMAoC$)? Can the parameter $k$ for which Theorem~\ref{thm:main_theorem} holds be reduced? Whereas the classical result only assumes $k \geq 3$, the quantum analogue assumes $k \geq 15$.

Whether $\qf(k) < f(k)$ is one way of asking the question whether $\qsat$ is more \emph{restrictive} than $\sat$. The full version contains several other approaches that tackle this question. 
}%
\fv{%
 Is $\qf(k)=f(k)$? Similarly, is $l(k)=\ql(k)$? There are three techniques for proving a lower bound on $f(k)$: The first, given by Tovey~\cite{tovey1984simplified} showed that $f(k) \geq k$, based on Hall's marriage theorem. Using Hall's marriage theorem in the quantum setting also implies $\qf(k)\geq k$, but requires more complicated arguments~\cite{laumann09dimer}. The second technique uses the Lov\'asz Local Lemma, which was also generalized to the quantum setting (see Theorem~\ref{thm:bound_fk} and Theorem~\ref{thm:bound_qf}). The third technique uses the Lopsided Lov\'asz Local Lemma (see Theorem~\ref{thm:gebauer}) which asymptotically improves the previous bound by a multiplicative factor of $2$.
When one uses the (original) Lov\'asz Local Lemma, the relevant parameter is the degree of a variable. When one uses the Lopsided Lov\'asz Local Lemma, the relevant parameter is the degree of a literal (a variable or its negation). Therefore, for balanced $\ksat$ instances - instances in which each variable appears the same number of times as the negation of the variable - one gets the factor 2 improvement immediately. Gebauer et al. managed to show the same improvement also for non-balanced instances. The  Lopsided Lov\'asz Local Lemma  seems hard to generalize to the quantum setting, because it is unclear what could take the role of a literal. This may suggest that $\qf(k)<f(k)$: it may be possible to find non-satisfiable $\ksqsat$ instances in the (classically satisfiable) regime $ \lfloor \frac{2^{k}}{ek} \rfloor  \leq s <  \lfloor \frac{2^{k+1}}{e(k+1)} \rfloor$. 

The second open problem is to make the jump truly sudden (i.e. to prove that $\ksqsat[\qf(k)+1]$ is $\QMAoC$); and to decrease the parameter $k$ for which Theorem~\ref{thm:main_theorem} holds: whereas the classical result only assumes $k \geq 3$, the quantum analogue assumes $k \geq 15$. See the end of Section~\ref{sec:proof} for a comparison between the quantum and the classical proofs, where we also specify the reasons for these differences.

The last open question is much broader. The question whether $\qf(k) < f(k)$ is one way of asking whether $\qsat$ is more \emph{restrictive} than $\sat$. A different approach of asking this question is the following: We say that two $\ksat$ formulas have the same \emph{structure} if the only difference between them is the signs of the variables. For example, $\phi = (x_1 \vee \bar x_2) \wedge (\bar x_2 \vee x_3)$ and $\phi' =(\bar x_1 \vee \bar x_2) \wedge ( x_2 \vee x_3)$ have the same structure. This definition can also be generalized to the quantum setting, which allows us to compare classical $\sat$ and quantum rank-1 $\qsat$ instances. For example, $\phi$ has the same structure as the rank-1 $\qsat$ instance $Q=(\Pi_1,\Pi_2)$ where $\Pi_1=\ketbra{00}{00}_{1,2} \otimes I_{3}$, $\Pi_{2}= I_{1} \tensor \frac{1}{2}(\ket{00}+\ket{11})(\bra{00}+\bra{11})_{2,3}$. 

Is there a $\ksat$ formula $\phi$ with the following properties?

\begin{enumerate}
	\item All classical formulas which have the same structure as $\phi$ are satisfiable.
	\item There exists a rank-1 $\qsat$ instance with the same structure as $\phi$ that is unsatisfiable. 
\end{enumerate}
The answer to the above question is yes\footnote{Note that if the answer was no, it would imply that $\qf(k)=f(k)$.}, by the following example (see also Figure~\ref{fig:qsat_vs_sat}): $\phi = (x_1 \vee x_2) \wedge (x_2 \vee x_3) \wedge (x_1 \vee x_3) \wedge (x_1 \vee x_3)$. $Q = (\Pi_{1},\ldots,\Pi_{4})$ where $\Pi_{i} = \ket{\psi_{i}}\bra{\psi_{i}}$ (tensored with the identity on the remaining qubit) and 
\begin{align*}
\ket{\psi_{1}}&= \frac{1}{\sqrt{2}}( \ket{01} - \ket{10} )\\
\ket{\psi_{2}}&= \frac{1}{\sqrt{2}}( \ket{01} - \ket{10} )\\
\ket{\psi_{3}}&=  \ket{00}\\
\ket{\psi_{4}}&= \ket{11}.
\label{eq:}
\end{align*}
In fact, by sampling the quantum constraints at random from the Haar measure, the instance would remain unsatisfiable with probability 1, by the Geometrization Theorem given in~\cite{laumann2009phase}.

Therefore,  $\qsat$ is more restrictive than $\sat$ in this sense, and the ``quantumness'' of the constraints makes them harder to satisfy. One may ask further questions: How generic is this phenomenon? Does entanglement play a central role, or are tensor product constraints as restrictive as entangled constraints? Can this restrictive property of $\qsat$ with respect to $\sat$ be exploited for a computational or another task?

\tikzset{My Style/.style={}}
	\tikzstyle{rnodes} = [circle,fill=blue!20]
\begin{figure}%
\centering
\begin{subfigure}{.5\textwidth}
  \centering
\begin{tikzpicture}
  [scale=.5,auto]
  \node [rnodes] (n1) at (-5,0) {$q_{1}$};
  \node [rnodes] (n2) at (0,7)  {$q_{2}$};
  \node [rnodes] (n3) at (5,0)  {$q_{3}$};

  \foreach \from/\to/\labl in { n1/n2/$\ket{10}$, n2/n3/$\ket{11}$}
    \draw (\from) to node {\labl} (\to);

	\path (n1) edge [bend left] node [My Style] {$\ket{01}$}(n3);
	\path (n1) edge [bend right] node {$\ket{10}$}(n3);

\end{tikzpicture}
  \caption{}
  \label{fig:sub1}
\end{subfigure}%
\begin{subfigure}{.5\textwidth}
  \centering
  \begin{tikzpicture}
  [scale=.5,auto]
  \node [rnodes] (n1) at (-5,0) {$q_{1}$};
  \node [rnodes] (n2) at (0,7)  {$q_{2}$};
  \node [rnodes] (n3) at (5,0)  {$q_{3}$};

  \foreach \from/\to/\labl in { n1/n2/$ \frac{1}{\sqrt{2}}( \ket{01} - \ket{10} )$, n2/n3/$ \frac{1}{\sqrt{2}}( \ket{01} - \ket{10} )$}
    \draw (\from) to node {\labl} (\to);

	\path (n1) edge [bend left] node [My Style] {$\ket{00}$}(n3);
	\path (n1) edge [bend right] node {$\ket{11}$}(n3);

\end{tikzpicture}

  \caption{}
  \label{fig:sub2}
\end{subfigure}
\caption{The above two examples are $2 \mhyphen \qsat$ instances. The qubits are the nodes, and the rank-1 projectors are the edges. The state on which the rank-1 projectors project is given beside each edge. The two instances have the same structure. (\subref{fig:sub1}) is in the computational basis (hence, it is equivalent to a $2-\sat$ formula), and (\subref{fig:sub2}) is not. The assignment $\ket{000}$ satisfies (\subref{fig:sub1}). It can be verified that every $2-\sat$ formula with the same structure as (\subref{fig:sub1}) is satisfiable. On the other hand, (\subref{fig:sub2}) is unsatisfiable. }
\label{fig:qsat_vs_sat}%
\end{figure}
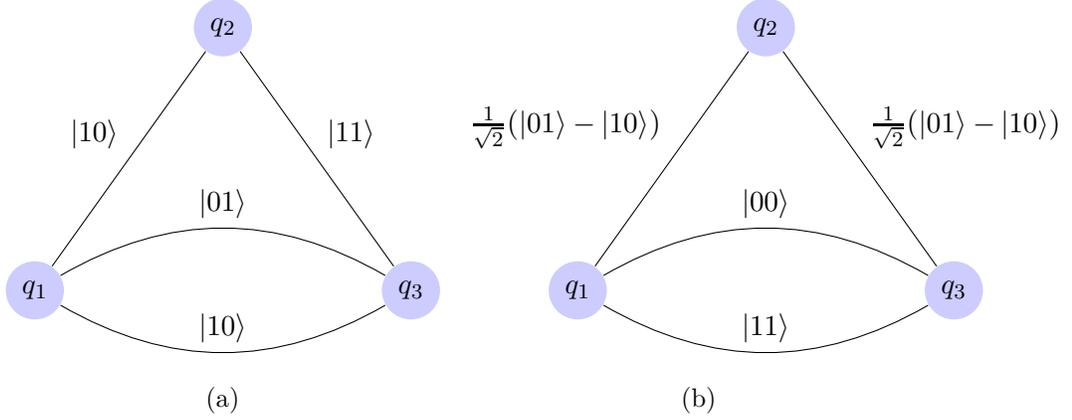
}

\section{Preliminaries}
\begin{definition}[$\QMAo,\QMA$] A language $L \in \QMA(s,c)$ if for every $x$ there exists a uniformly generated polynomial quantum circuit $V_{x}$ such that:
\begin{itemize}
\item(completeness) $x \in L \Rightarrow \exists \ket{\psi} $ such that $\Pr(V_{x} \text{ accepts } \ket{\psi}) \geq c $.
\item(soundness) $x \notin L \Rightarrow \forall \ket{\psi}, \  \Pr(V_{x} \text{ accepts } \ket{\psi}) \leq s $.

$\QMAo \equiv \QMA(\frac{1}{3},1)$, and $\QMA \equiv \QMA(\frac{1}{3},\frac{2}{3})$.
\label{def:qma}
\end{itemize}

\end{definition}
\begin{definition}[$\kqsat$ \cite{bravyi2006efficient}]
Input: An integer $n$, a real number $\epsilon = \Omega(1/n^{\alpha})$ for some constant $\alpha$, and a family of Hermitian projectors $\{\Pi_{1},\dots,\Pi_{m}\}$ where for all $i$, $\Pi_{i}$ acts non trivially on at most $k$ qubits out of the $n$ qubits. Let $Q=\sum_{i=1}^{m}\Pi_{i}$, and $\lambda_{0}(Q)$ be the minimal eigenvalue of $Q$.

Promise: Either $\lambda_{0}(Q)=0$ (in which case the instance is said to be satisfiable) or $\lambda_{0}(Q) \geq \epsilon$ (in which case the instance is unsatisfiable).

Problem: Decide which case it is.

An important parameter for the qPCP is $e_{0}(Q)=\lambda_{0}(Q)/m$.
\label{def:qsat}
\end{definition}

\begin{definition}[$\klh$ \cite{kitaev2002classsical}]
Input: An integer $n$, $a,b \in \mathbb{R}$ such that $|a-b| = \Omega(1/n^{\alpha})$ for some constant $\alpha$, and a family of positive Hermitian operators $\{H_{1},\dots,H_{m}\}$ where for all $i$, $H_{i}$ acts non trivially on at most $k$ qubits out of the $n$ qubits, and $||H_{i}||_{2}\leq 1$. Let $H=\sum_{i=1}^{m}H_{i}$, and $\lambda_{0}(H)$ be the minimal eigenvalue of $H$.

Promise: Either $\lambda_{0}(H)\leq a $ or $\lambda_{0}(H) \geq b$.

Problem: Decide which case it is.
\label{def:klh}
\end{definition}




\begin{definition}[$d$-state $\qsat$]
Defined  as $\qsat$, except the projectors act on $d$-level quantum systems (qudits), instead of qubits (qudits with $2$ levels).
\label{def:d_state_qsat}
\end{definition}

\begin{definition}[$1$-DIM $\qsat$]
Defined  as $2-\qsat$, with the additional requirement that each projector can act non-trivially only on the $j^{\text{\tiny th}}$ and $j+1^{\text{\tiny th}}$ qubits for some $j \in \{1,\ldots,n-1 \}$.
\label{def:1_d_qsat}
\end{definition}

\begin{definition}[rank-$1$ $\kqsat$]
Defined  as $\qsat$, except the rank of each projector is one (on the $k$ qubits that the projector acts on nontrivially). In other words, each projector must have the form, $\Pi = \ketbra{\psi}{\psi} \tensor I $ for some quantum state $\ket{\psi}$.
\label{def:rank_one_qsat}
\end{definition}


\begin{definition}[Regularity and Degree] Given a $\qsat$ instance $Q$ over n qubits, for $1 \leq j \leq n$ let $\Delta(j)$ be the number of projectors which act non-trivially on the $j^{th}$ qubit. We say that the instance is regular if $\forall i , j \quad  \Delta(i) = \Delta(j)$.
Let the degree of the instance be defined by:
\[ \Delta(Q)= \max_{1 \leq j \leq n} \Delta(j). \]
Regularity and degree are defined in a similar manner also for the \lh\ problem.
\label{def:regularity_degree}
\end{definition}

\begin{definition}[$\ksqsat$] Defined  as rank-$1$ \qsat, with the additional requirement that $\Delta(Q) \leq s$.
\label{def:ksqsat}
\end{definition}



\subsection{Notation}
Unless otherwise stated, all vectors $\ket{\psi}$ are normalized: $|\braket{\psi}{\psi}|=1$.  Given a multi qubit state of system $A$ and $B$ we use a subscript to denote the two systems. If the state of the $A$ system is $\ket{\alpha}$ and the $B$ system is $\ket{\beta}$, then, the entire state will be denoted as $\ket{\alpha}_A \tensor \ket{\beta}_B$. By abuse of notation, we treat a $\kqsat$ instance $Q$, both as the instance, and as the sum of all the projectors: $Q=\sum_{i} \Pi_{i}$.

\section{Proof of the main Theorem}
\label{sec:proof}
In this section we prove Theorem~\ref{thm:main_theorem}. It is already known that $\kqsat \in \QMAo$ for any constant $k$~\cite{bravyi2006efficient}, which implies that also $\ksqsat \in \QMAo$. 

Our starting point is the following result.
\begin{theorem}[\cite{aharonov_power_2007,nagaj2008local}] $1$-DIM $12$-state $\qsat$ is $\QMAoC$.
\end{theorem}
Therefore, we need to show a reduction from $1$-DIM $12$-state $\qsat$ to $\ksqsat[\qf(k)+2]$, for any $k\geq 15$. 

Note that the above result involves $12$-state qudits (see Def.~\ref{def:d_state_qsat}), whereas $\qf(k)$ is defined for ($2$-state) qubits (see Eq.~\eqref{def:qfk}). Therefore, we transform each qudit of dimension $12$ to $4$ qubits. Thus, the interaction becomes $k'=8$ local. Furthermore, we replace each (not necessarily rank-1) projector with at most $2^8 - 1 = 255$ rank-1 projectors, which we denote as $Q$. The instance $Q$ is $k'=8$ local, and has degree
\begin{equation}
\Delta(Q) \leq 510.
\label{eq:delta_Q}
\end{equation}

We now construct a $\kqsat$ instance denoted $T$. For each $k'$-local projector $\Pi \in Q$, we add $k-k'$ dummy qubits; we replace $\Pi$ with the following $k$ local projector $\Pi'$, defined as follows:
\begin{equation}
\Pi'=\Pi \tensor \ket{0}\bra{0}_{dummy_{1}}\tensor \ldots \tensor \ket{0}\bra{0}_{dummy_{k-k'}}.
\label{eq:pi_prime}
\end{equation}
At this point we wish to enforce all the dummy qubits to be in the $\ket{0}$ state: this would imply that the satisfiability of the instance $Q$ has not changed due to the transformation $\Pi \rightarrow \Pi'$. This enforcing gadget will be denoted $S$, and for each dummy qubit $q$, we add all the constraints of $S(q)$ and its qubits, denoted as the ancilla qubits. In total, we add $(k-k')|S|$ qubits for each projector $\Pi'$, where $|S|$ is the number of qubits in $S$.

We are now ready to describe the enforcing gadget $S$ and its properties. A $\qsat$ instance $Q$ is {\em minimal} if it is unsatisfiable, and for every projector $\Pi \in Q$, the instance $Q \setminus \{\Pi\}$ is satisfiable. There exists a minimal $\ksqsat[\qf(k)+1]$ instance $R$: by the definition of $\qf(k)$, there exists a non satisfying $\ksqsat[\qf(k)+1]$ instance; iteratively, we remove projectors if after removing them, the instance remains non-satisfiable. Let $\Lambda \in R$ be a rank-1 projector which acts non-trivially on the first qubit and on $k-1$ other qubits, which we denote as the set $A$.  Given $\Lambda = \ketbra{\psi}{\psi}$, using the Schmidt decomposition, we can write $\ket{\psi}=\sum_{i=1}^2 \sqrt{p_i} \ket{\alpha_i}_{\text{first qubit}} \tensor \ket{\beta_i}_A$, where $p_1+p_2=1$, and $p_{1},p_{2}\geq 0$. For $i=1,2$ let $\Lambda_i= \ketbra{\beta_i}{\beta_i}_A$. Note that the two projectors $\Lambda_1, \Lambda_2$ are more restrictive than $\Lambda$ in the following sense: for every state $\ket{\varphi}$, 

\begin{equation}
\bra{\varphi} \sum_{i=1}^2 \Lambda_i \ket{\varphi} \geq \bra{\varphi} \Lambda \ket{\varphi}.
\label{eq:lambda}
\end{equation} 

We replace the projector $\Lambda$ with two other $k$-local projectors $\tilde{\Lambda}_i = \Lambda_i \tensor \ketbra{1}{1}_{dummy}$ for $i=1,2$. We denote this enforcing gadget as the instance $S$.

\begin{lemma}
The $\kqsat$ instance $S$ has the following properties:
\begin{enumerate}
\item $S$ is satisfiable by a state of the form $\ket{\psi}\tensor\ket{0}_{dummy}$. \label{it:satisfiable}
\item There exists a constant $c_k$ (which can only depend on $k$) such that for all states $\ket{\phi} = \ket{\psi} \tensor \ket{1}_{dummy}$, $\bra{\phi} S \ket{\phi} \geq c_k$. \label{it:min_energy}
\item $\Delta(S)\leq \qf(k)+2$. \label{it:delta_s}
\end{enumerate}
\label{le:S_properties}
\end{lemma}

\begin{proof}
\ref{it:satisfiable}. Since $R$ is minimal, $R \setminus \{\Lambda \}$ is satisfiable, and let $\ket{\psi}$ be a satisfying state for $R \setminus \{\Lambda \}$. The state $\ket{\psi} \tensor \ket{0}_{dummy}$ also satisfies $\tilde{\Lambda}_1, \tilde{\Lambda}_2$, and therefore satisfies $S$. 

\ref{it:min_energy}. Since $R$ is unsatisfiable, there exists a constant $c_k$ (which only depends on $k$) such that for every state $\ket{\psi}$, $\bra{\psi} R \ket{\psi} \geq c_k$. Note that for every state $\ket{\phi}=\ket{\psi}\tensor \ket{1}_{dummy}$, $\bra{\phi}  S \ket{\phi}  \geq \bra{\psi} R \ket{\psi}$: all the constraints that do not involve the dummy qubit are not affected, and since $\bra{\phi}\sum_{i=0}^1 \tilde{\Lambda}_i \ket{\phi} = \bra{\psi} \sum_{i=0}^1 \Lambda_i \ket{\psi} \geq \bra{\psi} \Lambda \ket{\psi}$, where the last inequality follows from Eq. \eqref{eq:lambda}.


\ref{it:delta_s}. $\Delta(R) \leq \qf(k)+1$. By replacing $\Lambda$ with $\Lambda_{1}$ and $\Lambda_{2}$, the degree of each qubit in $\Lambda$ increases by at most 1, and the degree of the dummy qubit is 2. Therefore, $\Delta(S)\leq \qf(k)+2$.
\end{proof}

\begin{lemma}
Let $E$ be the minimal eigenvalue of $Q$, and let $E'$ be the minimal eigenvalue of $T$. If $E \leq c_{k}$, then $E'=E$, otherwise, $E' \geq c_{k}$,
where $c_{k}$ is the constant defined in Lemma \ref{le:S_properties}.
\label{le:min_eigenvalue}
\end{lemma}
\begin{proof}
We can decompose the entire vector space to a direct sum of subspaces based on the state of the dummy qubits in the computational basis. These subspaces are invariant under $T$ because all the projectors in $T$ commute with $\sigma_{z} = \begin{pmatrix}  1 & 0 \\  0 & -1 \end{pmatrix}$. In the subspace in which the state of the dummy qubits is $\ket{x}$, where $x \in \{0,1\}^{m}$, and $m$ is the total number of dummy qubits, 
\[ (\bra{\psi} \tensor  \bra{x}) T (\ket{\psi} \tensor \ket{x}) \geq c_{k} \cdot \text{ham}(x), \]
where $ \text{ham}(x)$ is the Hamming weight of $x$. 
This inequality follows from Lemma~\ref{le:S_properties}.\ref{it:min_energy}.

Every state of the form $\ket{\Omega}=\ket{\alpha}_{work}\tensor\ket{0^{m}}_{dummy}\tensor \ket{\phi}_{ancilla}$ satisfies
\begin{equation}
\bra{\Omega}T \ket{\Omega} \geq \bra{\alpha} Q \ket{\alpha} \geq E.
\label{eq:T_G_geq_E}
\end{equation} 
Since $T$ is invariant in these subspaces, Eq.~\eqref{eq:T_G_geq_E} also holds for superposition of states of this form, i.e. states of the form $\ket{\Omega}=\sum_i a_i \ket{\alpha_i} \tensor \ket{0^m} \tensor \ket{\phi_i}$.

Let $\ket{\Omega_{0}}=\ket{\alpha_{0}}_{work}\tensor\ket{0^{m}}_{dummy}\tensor (\ket{\psi}^{\otimes m})_{ancilla}$, where $\ket{\alpha_{0}}$ is an eigenvector of $Q$ with eigenvalue $E$, and $\ket{0} _{dummy}\otimes \ket{\psi}_{ancilla}$ is a satisfying state for $S$, which is guaranteed to exist by Lemma~\ref{le:S_properties}.\ref{it:satisfiable}. This state satisfies
\[ \bra{\Omega_{0}}T \ket{\Omega_{0}} = \bra{\alpha_{0}} Q \ket{\alpha_{0}} = E.\]

 \onote{improve proof}
\end{proof}

We are now ready to complete the proof of Theorem~\ref{thm:main_theorem}. We reduce the instance $Q$ to the instance $T$, where we use $\epsilon'=\min\{ \epsilon,c_{k} \}$, where $\epsilon$ is the parameter for original instance $Q$, and $c_k$ is the parameter from Lemma~\ref{le:S_properties}. By Lemma~\ref{le:min_eigenvalue}, the minimum eigenvalue of $T$ is $0$ if $Q$ is a ``yes'' instance, and at least $\min\{ \epsilon,c_{k} \}$ if $Q$ is a ``no'' instance.
The locality of $T$ is indeed $k$.
We claim that $\Delta(T) \leq \max \{\qf(k)+2,\ 510\}$.
$\Delta(Q)\leq 510$ by Eq.~\eqref{eq:delta_Q}, hence, the degree of the qubits which originate from $Q$ is at most $510$. The degree of each dummy qubit is $3$.
Since $\Delta(S)\leq \qf(k)+2$ by Lemma~\ref{le:S_properties}.\ref{it:delta_s}, the degree of the ancilla qubits which originate from the enforcing gadgets S is at most $\qf(k)+2$.  
Since $k \geq 15$ (by the assumption of Theorem~\ref{thm:main_theorem}), it can be verified using Theorem~\ref{thm:bound_qf} that $\max\{\qf(k)+2,\ 510\} = \qf(k)+2$, therefore, $\Delta(T)\leq \qf(k)+2$  which completes the proof of Theorem~\ref{thm:main_theorem}. \qed

\paragraph{Comparison between the quantum and the classical proofs.} We now compare the above quantum proof with the classical proof of Theorem~\ref{thm:kratochvil}. There are two main steps in both proofs. The first, is to decrease $k$ and the degree to the minimal value possible. The second, is to increase $k$ without increasing the degree much above $f(k)$ in the classical setting and $\qf(k)$ in the quantum setting. 

In the first step of the classical setting, we start with a $3\mhyphen\sat$ instance (which is $\NPC$), and we replace each variable $x$ that appears $r$
 times with $x_{1}, \ldots, x_{r}$, and we add additional clauses that enforce that $x_{1}=\ldots = x_{r}$, while keeping the degree below a small constant. Imposing equality between qubits is not well defined in the quantum setting, where it is a common barrier (for example, in quantum error correcting codes). 
For this reason, we use a completely different approach:
The final Hamiltonian of the $\QMA$-hardness reduction of $\qsat$ on a line has bounded degree and locality, which are exactly the properties that are needed.  
The bottleneck for proving Theorem~\ref{thm:main_theorem} for $k$ smaller than $15$ is due to the properties of the first step: we already start with $k'=8$, and $\Delta(Q)\leq 510$.
These parameters could potentially be optimized using a different (standard) $\QMA_{1}$-completeness constructions, or a tailor made construction that minimizes $k'$ and $\Delta(Q)$. 

The second step is very similar in spirit in both the classical and quantum proofs, although it contains one crucial difference. In the classical case, one can replace a $k$ clause by a $k-1$ clause which is more restrictive, by removing an arbitrary variable from the $k$ clause\footnote{Note that we assume that each $k$ clause contains exactly $k$ \emph{different} variables. Otherwise, this statement would not hold.}. In our case, we have to replace a rank-1 $k$ local projector with \emph{two} rank-1 $k-1$ local projectors which are more restrictive (see Eq.~\eqref{eq:lambda}). The effect of this difference is that in the classical case it is $\NPH$ to decide $\kssat[f(k)+1)]$, while in the quantum case only $\ksqsat[\qf(k)+2]$ is $\QMAoH$.

\section{Acknowledgments}
The author wish to thank Martin Schwarz for suggesting the connection to the qPCP conjecture, and  Dorit Aharonov, Itai Arad, and Yosi Atia for fruitful discussions.
\bibliographystyle{alphaabbr}	
\bibliography{almost_sudden_jump_in_quantum_complexity}

\newcommand{\etalchar}[1]{$^{#1}$}
\begin{thebibliography}{GMSW09}

\bibitem[AAV13]{aharonov2013guest}
D.~Aharonov, I.~Arad, and T.~Vidick.
\newblock Guest column: the quantum {PCP} conjecture.
\newblock {\em ACM SIGACT News}, 44(2):47--79, 2013.
\newblock Arxiv preprint arXiv:1309.7495.

\bibitem[AGIK07]{aharonov_power_2007}
D.~Aharonov, D.~Gottesman, S.~Irani, and J.~Kempe.
\newblock The power of quantum systems on a line.
\newblock In {\em Proceedings of the 48th Annual {IEEE} Symposium on
  Foundations of Computer Science}, pages 373--383. {IEEE} Computer Society,
  2007.

\bibitem[AKS12]{ambainis2010quantum}
A.~Ambainis, J.~Kempe, and O.~Sattath.
\newblock A quantum {L}ov\'asz local lemma.
\newblock {\em J. ACM}, 59(5):24:1--24:24, November 2012.

\bibitem[AN02]{aharonov2002quantum}
D.~Aharonov and T.~Naveh.
\newblock Quantum {NP}-a survey.
\newblock {\em arXiv preprint quant-ph/0210077}, 2002.

\bibitem[AS04]{alon2004probabilistic}
N.~Alon and J.~Spencer.
\newblock {\em {The probabilistic method}}.
\newblock Wiley-Interscience, 2004.

\bibitem[BDLT08]{bravyi2008quantum}
S.~Bravyi, D.~P. DiVincenzo, D.~Loss, and B.~M. Terhal.
\newblock Quantum simulation of many-body {H}amiltonians using perturbation
  theory with bounded-strength interactions.
\newblock {\em Physical review letters}, 101(7):070503, 2008.

\bibitem[BH13]{brandao2013product}
F.~G. Brand{\~a}o and A.~W. Harrow.
\newblock Product-state approximations to quantum ground states.
\newblock In {\em Proceedings of the 45th annual ACM symposium on Symposium on
  theory of computing}, pages 871--880. ACM, 2013.

\bibitem[BKS03]{berman2003approximation}
P.~Berman, M.~Karpinski, and A.~D. Scott.
\newblock Approximation hardness and satisfiability of bounded occurrence
  instances of {SAT}.
\newblock {\em Electronic Colloquium on Computational Complexity (ECCC)}, 2003.

\bibitem[Boo12]{bookatz2012qma}
A.~D. Bookatz.
\newblock {QMA}-complete problems.
\newblock {\em arXiv preprint arXiv:1212.6312}, 2012.

\bibitem[Bra06]{bravyi2006efficient}
S.~Bravyi.
\newblock {Efficient algorithm for a quantum analogue of 2-SAT}.
\newblock {\em Arxiv preprint quant-ph/0602108}, 2006.

\bibitem[EL75]{erdos1975problems}
P.~Erd{\H{o}}s and L.~{L}ov\'{a}sz.
\newblock {Problems and results on 3-chromatic hypergraphs and some related
  questions}.
\newblock {\em Infinite and finite sets}, 2:609--627, 1975.

\bibitem[ER08]{eldard2008quantum}
L.~Eldar and O.~Regev.
\newblock Quantum {SAT} for a {Q}utrit-{C}inquit pair is {QMA} 1-{C}omplete.
\newblock In {\em Automata, Languages and Programming}, volume 5125 of {\em
  Lecture Notes in Computer Science}, pages 881--892. Springer Berlin
  Heidelberg, 2008.

\bibitem[ES91]{erdos1991lopsided}
P.~Erdos and J.~Spencer.
\newblock Lopsided {L}ov{\'a}sz local lemma and latin transversals.
\newblock {\em Discrete Applied Mathematics}, 30(2):151--154, 1991.

\bibitem[GMSW09]{gebauer2009lovasz}
H.~Gebauer, R.~Moser, D.~Scheder, and E.~Welzl.
\newblock The {L}ov{\'a}sz local lemma and satisfiability.
\newblock {\em Efficient Algorithms}, pages 30--54, 2009.

\bibitem[GN13]{gosset2013quantum}
D.~Gosset and D.~Nagaj.
\newblock Quantum 3-{SAT} is{ QMA}1-complete.
\newblock {\em arXiv preprint arXiv:1302.0290}, 2013.

\bibitem[GST11]{gebauer2010local}
H.~Gebauer, T.~Szab{\'o}, and G.~Tardos.
\newblock The local lemma is tight for {SAT}.
\newblock In D.~Randall, editor, {\em SODA}, pages 664--674. SIAM, 2011.

\bibitem[HS05]{hoory2005computing}
S.~Hoory and S.~Szeider.
\newblock Computing unsatisfiable k-{SAT} instances with few occurrences per
  variable.
\newblock {\em Theoretical Computer Science}, 337(1):347--359, 2005.

\bibitem[KST93]{kratochvil1993one}
J.~Kratochv{\'\i}l, P.~Savick{\`y}, and Z.~Tuza.
\newblock One more occurrence of variables makes satisfiability jump from
  trivial to {NP}-complete.
\newblock {\em SIAM Journal on Computing}, 22(1):203--210, 1993.

\bibitem[KSV02]{kitaev2002classsical}
A.~Kitaev, A.~H. Shen, and M.~N. Vyalyi.
\newblock {\em Classsical and quantum computation}.
\newblock Number~47 in Graduate studies in mathematics. American Mathematical
  Soc., 2002.

\bibitem[LLM{\etalchar{+}}10]{laumann09dimer}
C.~R. Laumann, A.~L{\"a}uchli, R.~Moessner, A.~Scardicchio, and S.~Sondhi.
\newblock Product, generic, and random generic quantum satisfiability.
\newblock {\em Physical Review A}, 81(6):062345, 2010.

\bibitem[LMSS09]{laumann2009phase}
C.~Laumann, R.~Moessner, A.~Scardicchio, and S.~Sondhi.
\newblock Phase transitions and random quantum satisfiability.
\newblock {\em Arxiv preprint arXiv:0903.1904}, 2009.

\bibitem[Nag08]{nagaj2008local}
D.~Nagaj.
\newblock Local hamiltonians in quantum computation.
\newblock {\em arXiv preprint arXiv:0808.2117}, 2008.

\bibitem[Osb12]{osborne2012hamiltonian}
T.~J. Osborne.
\newblock Hamiltonian complexity.
\newblock {\em Reports on Progress in Physics}, 75(2):022001, 2012.

\bibitem[Tov84]{tovey1984simplified}
C.~A. Tovey.
\newblock A simplified {NP}-complete satisfiability problem.
\newblock {\em Discrete Applied Mathematics}, 8(1):85--89, 1984.

\end{thebibliography}
\end{document}